\documentclass[conference,letterpaper]{IEEEtran}

\usepackage{graphicx}
\usepackage{color}
\usepackage{pgf, tikz, pgfplots}
\usepackage{siunitx}
\pgfplotsset{compat=1.17} 
\usepackage{amssymb} 
\usepackage{bm}
\usepackage{acronym}
\usepackage{tikz}
\usepackage{mathrsfs}
\usepackage{enumerate}

\usepackage[utf8]{inputenc} 
\usepackage[T1]{fontenc}
\usepackage{url}
\usepackage{ifthen}
\usepackage{cite}
\usepackage[cmex10]{amsmath} %

\usetikzlibrary{arrows,calc,fit,matrix,positioning,spy}
\tikzset{myblock/.style={rectangle, draw, thin, minimum width=0.6cm, minimum height=0.6cm},font=\footnotesize,align=center}%
\tikzset{mywideblock/.style={rectangle, draw, thin, minimum width=1cm, minimum height=0.6cm},font=\footnotesize,align=center}%
\newcommand{\myline}[2]{
\path(#1.east) --(#2.west)  coordinate[pos=0.4](mid);
\draw[-latex] (#1.east) -| (mid) |- (#2.west);
}

\usepackage{cite}
\usepackage{amsmath,amssymb,amsfonts}

\usepackage{dsfont}
\usepackage{algorithmic}
\usepackage{graphicx}
\usepackage{textcomp}
\usepackage{xcolor}
\def\BibTeX{{\rm B\kern-.05em{\sc i\kern-.025em b}\kern-.08em
    T\kern-.1667em\lower.7ex\hbox{E}\kern-.125emX}}

    \usepackage{amsfonts}     
\usepackage{graphicx}
\usepackage{color}
\usepackage{pgf, tikz, pgfplots}
\usepackage{siunitx}
\pgfplotsset{compat=1.17} 
\usepackage{amssymb} 
\usepackage{bm}
\usepackage{bbm}
\usepackage{amsthm}
\usepackage{acronym}
\usepackage{tikz}
\usepackage{mathrsfs}
\usepackage{enumerate}
\usetikzlibrary{arrows,calc,fit,matrix,positioning,shapes,shadows,trees,mindmap,tikzmark,arrows.meta,angles,quotes,babel,spy}

\usepackage{booktabs}

\DeclareMathOperator*{\mathboxplus}{\boxplus}

\DeclareMathOperator*{\mAut}{\mathrm{GAut}(\mathcal{C})}
\DeclareMathOperator*{\mEnd}{\mathrm{End}(\mathcal{C})}
\DeclareMathOperator*{\mMatrixEnd}{\mathcal{T}_E(\mathcal{C})}

\definecolor{kit-green100}{rgb}{0,.59,.51}

\definecolor{kit-green100}{rgb}{0,.59,.51}
\definecolor{kit-green70}{rgb}{.3,.71,.65}
\definecolor{kit-green50}{rgb}{.50,.79,.75}
\definecolor{kit-green30}{rgb}{.69,.87,.85}
\definecolor{kit-green15}{rgb}{.85,.93,.93}
\definecolor{KITgreen}{rgb}{0,.59,.51}

\definecolor{KITpalegreen}{RGB}{130,190,60}
\colorlet{kit-maigreen100}{KITpalegreen}
\colorlet{kit-maigreen70}{KITpalegreen!70}
\colorlet{kit-maigreen50}{KITpalegreen!50}
\colorlet{kit-maigreen30}{KITpalegreen!30}
\colorlet{kit-maigreen15}{KITpalegreen!15}

\definecolor{KITblue}{rgb}{.27,.39,.66}
\definecolor{kit-blue100}{rgb}{.27,.39,.67}
\definecolor{kit-blue70}{rgb}{.49,.57,.76}
\definecolor{kit-blue50}{rgb}{.64,.69,.83}
\definecolor{kit-blue30}{rgb}{.78,.82,.9}
\definecolor{kit-blue15}{rgb}{.89,.91,.95}

\definecolor{KITyellow}{rgb}{.98,.89,0}
\definecolor{kit-yellow100}{cmyk}{0,.05,1,0}
\definecolor{kit-yellow70}{cmyk}{0,.035,.7,0}
\definecolor{kit-yellow50}{cmyk}{0,.025,.5,0}
\definecolor{kit-yellow30}{cmyk}{0,.015,.3,0}
\definecolor{kit-yellow15}{cmyk}{0,.0075,.15,0}

\definecolor{KITorange}{rgb}{.87,.60,.10}
\definecolor{kit-orange100}{cmyk}{0,.45,1,0}
\definecolor{kit-orange70}{cmyk}{0,.315,.7,0}
\definecolor{kit-orange50}{cmyk}{0,.225,.5,0}
\definecolor{kit-orange30}{cmyk}{0,.135,.3,0}
\definecolor{kit-orange15}{cmyk}{0,.0675,.15,0}

\definecolor{KITred}{rgb}{.63,.13,.13}
\definecolor{kit-red100}{cmyk}{.25,1,1,0}
\definecolor{kit-red70}{cmyk}{.175,.7,.7,0}
\definecolor{kit-red50}{cmyk}{.125,.5,.5,0}
\definecolor{kit-red30}{cmyk}{.075,.3,.3,0}
\definecolor{kit-red15}{cmyk}{.0375,.15,.15,0}

\definecolor{KITpurple}{RGB}{160,0,120}
\colorlet{kit-purple100}{KITpurple}
\colorlet{kit-purple70}{KITpurple!70}
\colorlet{kit-purple50}{KITpurple!50}
\colorlet{kit-purple30}{KITpurple!30}
\colorlet{kit-purple15}{KITpurple!15}

\definecolor{KITcyanblue}{RGB}{80,170,230}
\colorlet{kit-cyanblue100}{KITcyanblue}
\colorlet{kit-cyanblue70}{KITcyanblue!70}
\colorlet{kit-cyanblue50}{KITcyanblue!50}
\colorlet{kit-cyanblue30}{KITcyanblue!30}
\colorlet{kit-cyanblue15}{KITcyanblue!15}

\newcommand\blfootnote[1]{%
    \begingroup
    \renewcommand\thefootnote{}\footnote{#1}%
    \addtocounter{footnote}{-1}%
    \endgroup
}

\usepackage[colorlinks=true,bookmarks=false,citecolor=black,urlcolor=black,linkcolor=black]{hyperref} 

\newtheorem{theorem}{Theorem}
\newtheorem{proposition}{Proposition}

\begin{document}
\title{%
Endomorphisms of Linear Block Codes\\
}

\author{\IEEEauthorblockN{Jonathan Mandelbaum, Sisi Miao, Holger Jäkel, and Laurent Schmalen}
\IEEEauthorblockA{Communications Engineering Lab, Karlsruhe Institute of Technology (KIT), 76131 Karlsruhe, Germany\\
\texttt{jonathan.mandelbaum@kit.edu}}
}

\maketitle

\begin{abstract}
The automorphism groups of various linear codes are extensively studied yielding insights into the respective code structure. 
This knowledge is used in, e.g., theoretical analysis and in improving decoding performance,
motivating the analyses of endomorphisms of linear codes.
In this work, we discuss the structure of the set of transformation matrices of code endomorphisms, defined as a generalization of code automorphisms, and provide an explicit
construction of a bijective mapping between the image of an endomorphism and its canonical quotient space. Furthermore, we introduce a one-to-one mapping between the set of transformation matrices of endomorphisms and a larger linear block code enabling the use of well-known algorithms for the search for suitable endomorphisms.
Additionally, we propose an approach to obtain unknown code endomorphisms based on automorphisms of the code.
Furthermore, we consider ensemble decoding as a possible use case for endomorphisms by introducing endomorphism ensemble decoding. Interestingly, EED can improve decoding performance when other ensemble decoding schemes are not applicable.%
\end{abstract}

\section{Introduction}

It is well-known that the automorphism group of a code provides valuable knowledge of the code's structure \cite{MacWilliamsSloane}. 
This knowledge is successfully used in different areas of coding theory, e.g., in the proof that Reed--Muller (RM) codes achieve the capacity on the erasure channel \cite{7862912}, in enumerating minimum weight codewords of polar codes \cite{algebraicpropertiesofpolarcodes}, and in novel ensemble decoding schemes improving performance in the short-block length regime \cite{AED_RMcodes,stuttgart_ldpc_aed, MBBP1}.
Such applications motivate the study of the automorphism group of various classical codes, e.g., Golay, Bose--Chaudhuri--Hocquenghem (BCH), and RM codes as well as modern codes like
polar codes\cite{MacWilliamsSloane,DBLP:journals/dcc/BergerC99,algebraicpropertiesofpolarcodes,Aut_PolarCodes_Geiselhart}. \blfootnote{This work has received funding from the 
German Federal Ministry of Education and Research (BMBF) within the project Open6GHub (grant agreement 16KISK010) and the European Research Council (ERC) under the European Union’s Horizon 2020 research and innovation programme (grant agreement No. 101001899).}

Typically, in coding theory, the automorphism group of a code is defined as the set of \emph{scaled permutations} mapping a code onto itself \cite{MacWilliamsSloane}. From a classical coding perspective, this definition is coherent, as the so-defined automorphisms preserve properties of the respective codeword, e.g., the Hamming weight. Yet, in linear algebra, the concept of automorphisms of a vector space is more general including all linear, bijective self-mappings. In \cite{mandelbaum2023generalized}, the present authors have suggested the use of generalized automorphisms for decoding.

In this work, we extend this perspective by considering the set of endomorphisms of linear block codes, i.e., we drop the necessity of bijectivity. We use the methodology described in \cite{mandelbaum2023generalized} to analyze some structural properties of the transformation matrices of code endomorphisms.
Then, we provide an explicit construction of the transformation matrix corresponding to the bijective mapping from the image of an endomorphism onto its canonical quotient space. 
Furthermore, we prove the existence of a one-to-one mapping between the set of endomorphism transformation matrices of a code and a higher dimensional linear block code and 
describe an approach to construct endomorphisms of a code given its automorphism group.
To the best of our knowledge, the usage of endomorphisms was not yet discussed within the area of coding theory.

Afterwards, we propose ensemble decoding based on endomorphisms and call the scheme endomorphism ensemble decoding (EED).
EED provides a flexible framework generalizing ensemble decoding schemes, e.g., multiple basis belief propagation (MBBP) \cite{MBBP1,MBBP2,MBBP3_withLeaking}, automorphism ensemble decoding (AED) \cite{AED_RMcodes,stuttgart_ldpc_aed,geiselhart2023ratecompatible}, and generalized AED (GAED) \cite{mandelbaum2023generalized}.
We show that there exist scenarios in which AED does not yield improvement compared to conventional decoding ~\cite[Corollarly 2.1]{AED_RMcodes} while endomorphisms can improve the decoding performance.

\section{Preliminaries}
In this work, we consider linear block codes over a finite field $\mathbb{F}_q$.
A linear block code $\mathcal{C}(n,k)$ is a $k$-dimensional subspace of the vector space $\mathbb{F}_q^n$, where the parameters $n\in \mathbb{N}$ and $k\in \mathbb{N}$ represent the block length and information length, respectively.
Contrary to typical coding conventions, vectors are column vectors as common in linear algebra.
 For simplicity, the parameters $(n,k)$ are omitted when obvious from the context.

A major advantage of linear block codes is that they can be efficiently constructed as the linear span of the columns of a generator matrix. 
Equivalently, they can be described as the null space of their parity\--check matrix (PCM) $\bm{H}\in \mathbb{F}_q^{(n-k)\times n}$, which we assume to be of full rank \cite{MacWilliamsSloane}, i.e,
$$\mathcal{C}\left(n,k\right)=\left\{\bm{x}\in \mathbb{F}_q^n:\bm{H} \bm{x} = \bm{0}\right\}=\mathrm{Null}(\bm{H}).$$

The automorphism group of a code is defined as the set of \emph{scaled permutations} that fix the code, i.e.,
\begin{equation*}
\mathrm{Aut}(\mathcal{C}):=  \!  
\left\{\! \pi^{(a)}: \mathcal{C}\!\to\!\mathcal{C}, \bm{x}\!\mapsto\!a\pi(\bm{x})\!:\! \pi\!\in S_n, a\in\mathbb{F}_q\!\setminus\!\{0\}    \right\},
\end{equation*}
where $a{\pi(\bm{x})=\begin{pmatrix}
ax_{\pi(1)},&\cdots &,ax_{\pi(n)}
\end{pmatrix}}^\mathsf{T}$ and $S_n$ denoting the symmetric group\cite{MacWilliamsSloane}.
In \cite{mandelbaum2023generalized}, a different view on the automorphism group of codes has been discussed, inspired by the more general definition of automorphisms used in linear algebra.
Therein, the (generalized) automorphism group of a linear code $\mathcal{C}$ is defined as
\begin{equation*}
\mathrm{GAut}(\mathcal{C}):=    
\left\{ \tau: \mathcal{C}\rightarrow\mathcal{C}: \text{ $\tau$ linear, $\tau$ bijective} \right\}.
\end{equation*}
Note that $\mathrm{Aut}(\mathcal{C})\subseteq \mathrm{GAut}(\mathcal{C})$ \cite{mandelbaum2023generalized} and that 
a linear mapping ${\tau:\mathcal{C}\rightarrow\mathcal{C}}$ can be represented by a transformation matrix ${\bm{T}\in \mathbb{F}_q^{n\times n}}$,
i.e., \begin{equation}\label{eq:linear_mapping}
\bm{x}_\tau:=\tau(\bm{x})=\bm{T}\bm{x}\in \mathcal{C}.
\end{equation}

\section{Endomorphisms of Codes} \label{sec:end_of_code}
In this section, we consider endomorphisms of linear codes.
To this end, note that set of endomorphisms of a vector space $\mathcal{C}$ is defined as all
linear self-mappings \cite{bhattacharya1994basic}, i.e., \vspace{-1mm}
\begin{equation*}
\mEnd:=    
\left\{ \tau: \mathcal{C}\rightarrow\mathcal{C}: \text{ $\tau$ linear} \right\},\vspace{-2mm}
\end{equation*}
dropping the necessity of $\tau$ being bijective.
Thus,
${\mathrm{GAut}(\mathcal{C})\subseteq \mEnd}$, i.e., every generalized automorphism is also an endomorphism.
Furthermore, we define the set of transformation matrices of endomorphisms 
\begin{equation*}
    \mMatrixEnd:= \left\{\bm{T}\in \mathbb{F}_q^{n\times n}:  \bm{x}\in \mathcal{C}\implies \bm{T}\bm{x}\in \mathcal{C}, \forall \bm{x}\in \mathcal{C} \right\}.
\end{equation*}
A mapping ${\tau\in\mEnd\setminus \mathrm{GAut}(\mathcal{C})}$ is called \emph{proper endomorphism}.
Note that a proper endomorphism is no longer injective. 
In accordance with the first isomorphism theorem \cite[3.107]{ladoneright}, we define a coset
\[
[\bm{x}]_{\tau}:=\bm{x}+\mathrm{Null}(\tau)
\]
consisting of all codewords that are mapped onto the same codeword $\bm{x}_\tau\in \mathcal{C}$ by $\tau$.
Then, the set of all cosets forms a vector space called quotient space of $\mathcal{C}$ modulo $\mathrm{Null}(\tau)$ denoted by $\mathcal{C}/\mathrm{Null}(\tau)$ \cite[3.103]{ladoneright}.

By examining the structural properties of transformation matrices in $\mMatrixEnd$, we arrive at two constructions of code endomorphisms. These
are described in Theorem \ref{theorem:structure_of_end} and Proposition \ref{proposition:superposition_auts}.
\subsection{Endomorphisms of Codes} 
Consider a linear code $\mathcal{C}$ and let $\tau$ be a proper endomorphism of the code.
Then, the endomorphism $\tau$ maps the code $\mathcal{C}$ onto a lower dimensional subcode $\tau(\mathcal{C})$ of $\mathcal{C}$, i.e., $\tau(\mathcal{C})\subset\mathcal{C} $. 
Next, we will investigate the structure of transformation matrices of endomorphisms, proposing Theorem~\ref{theorem:structure_of_end}.
Afterwards in Theorem~\ref{theorem:reconstruction_end}, we propose an explicit mapping $\rho:\tau(\mathcal{C})\rightarrow \mathcal{C}$ from the subcode onto $\mathcal{C}$ inverting the effect of the endomorphism $\tau$.

Inspired by \cite{mandelbaum2023generalized}, we define
\begin{align}
&\mathcal{A}:= \left\{\bm{M}\in \mathrm{GL}_n(\mathbb{F}_q):\bm{H}\bm{M}=\left[\bm{I}_{(n-k)\times(n-k)}\,\bm{0}_{(n-k)\times k} \right]\right \},\nonumber\\
&\mathcal{Z}
:=
\left\{\bm{Z} \in \mathbb{F}_q^{n\times n}:
\bm{Z}=
\begin{bmatrix}  \bm{C} &\bm{0}_{(n-k)\times k} \\
\bm{D}& \bm{E}\end{bmatrix}  \right\},\label{eq:set_z}
\end{align}
with $\mathrm{GL}_n(\mathbb{F}_q)$ denoting the general linear group and 
with ${\bm{C}\in\mathbb{F}_q^{(n-k)\times (n-k)}}$, $\bm{D}\in\mathbb{F}_q^{k\times (n-k)}$, and $\bm{E}\in\mathbb{F}_q^{k\times k}$.
Note that in contrast to \cite{mandelbaum2023generalized}, elements from $\mathcal{Z}$ are not required to be non-singular, i.e., matrices $\bm{E}$ and $\bm{C}$ are not necessarily of full rank. An element of the set $\mathcal{A}$ is called \emph{code characterization matrix} (CCM). 

Assuming ${\mathrm{rank}(\bm{E})=k}$ and an arbitrary, possibly rank deficient $\bm{C}$ still results in ${{\bm{A} \bm{Z}  \bm{A}^{-1} \in \mAut} \,\, \forall \bm{Z}\in\mathcal{Z}}, \bm{A}\in \mathcal{A}$, i.e., the mapping constitutes an automorphism of $\mathcal{C}$, albeit not defining an automorphism of $\mathbb{F}_q^n$.
This case has not been considered in \cite{mandelbaum2023generalized}.

Interestingly, allowing for an arbitrary $\bm{E}$, dropping the full rank requirement, yields the set of endomorphisms as stated in the following theorem.

\begin{theorem}\label{theorem:structure_of_end} Let $\mathcal{C}$ be a linear code with PCM $\bm{H}\in \mathbb{F}_q^{(n-k)\times n}$ of rank $n-k$ and CCM $\bm{A}\in \mathcal{A}$. 
Then, the mapping $\tau:\mathcal{C}\rightarrow\mathcal{C}$ with transformation matrix $\bm{T}$ is an element of $\mEnd$ if and only if there exists $\bm{Z}\in\mathcal{Z}$
such that $\bm{T}=\bm{A} \bm{Z}  \bm{A}^{-1}$.
\end{theorem}
\begin{proof}
The proof is similar to the proof of \cite[Theorem 1]{mandelbaum2023generalized}.
\end{proof}

The proof of Theorem \ref{theorem:structure_of_end} implies that, for an arbitrary but fixed CCM $\bm{A}$, all transformation matrices with equal $\bm{E}$ correspond to the same endomorphism $\tau \in \mEnd$.

Next, we derive the transformation matrix of a bijective mapping between its image $\mathrm{Im}(\tau)$ and $\mathcal{C}/\mathrm{Null}(\tau)$.
This analysis provides insights that are helpful to exploit endomorphisms in decoding.
First,  for arbitrary $\bm{M}\in \mathbb{F}_q^{k\times k}$ we define %
\begin{align*}
&\mathcal{J}(\bm{M}):=\left\{i\in\left\{1,\ldots,k\right\}:M_{ii}=0 \right\}\\
&\mathcal{J}^{\mathsf{c}}(\bm{M}):=\left\{1,\ldots,k\right\}\setminus\mathcal{J}(\bm{M})%
\end{align*}
and let
\begin{align*}
&\mathcal{M}:=\left\{\bm{M}\in \mathrm{LT}_k:  j \in \mathcal{J}(\bm{M})\implies M_{ij}=0,\, \forall i >j\right\}
\end{align*}
with $\mathrm{LT}_k$ denoting the set of $k\times k$ lower triangular matrices. Thus, the set $\mathcal{M}$ consists of lower triangular matrices that possess only zero entries in a column in which the respective element on the main diagonal is zero. 
Thereby, all non-zero columns of matrices in $\mathcal{M}$ are linearly independent. Note that for arbitrary ${\bm{M}\in \mathcal{M}}$, the set ${\left\{\bm{e}_j\in\mathbb{F}_q^k: j\in \mathcal{J}(\bm{M})\right\}}$, with $\bm{e}_j$ denoting the $j$-th canonical base vector,
consists of $|\mathcal{J}(\bm{M})|$ linearly independent vectors constituting a basis of $\mathrm{Null}(\bm{M})$.

\begin{theorem}\label{theorem:reconstruction_end}
Let $\mathcal{C}$ be a linear code with PCM $\bm{H}\in \mathbb{F}_q^{(n-k)\times n}$ of rank $n-k$ with CCM $\bm{A}\in \mathcal{A}$ 
and let ${\tau: \mathcal{C}\rightarrow\mathrm{Im}(\tau)\subseteq\mathcal{C}}$ be an endomorphism.
Then, there exists a bijective mapping ${\rho:\mathrm{Im}(\tau)\rightarrow \mathcal{C}/\mathrm{Null}(\tau)}$ with transformation matrix $\bm{R}=\bm{A} \bm{Z}_R  \bm{A}^{-1},\,\bm{Z}_R\in \mathcal{Z}$, such that for arbitrary $\bm{x}\in \mathcal{C}$:
\begin{equation}\label{eq:reconstruction}
\bm{R}\bm{T}\bm{x}\in [\bm{x}]_{\tau}.
\end{equation}
The matrix $\bm{R}$ is denoted as \emph{reconstruction matrix} and the coset $[\bm{x}]_{\tau}$ contains $q^s$ distinct codewords, with $s:=\mathrm{dim}(\mathrm{Null}(\tau))$ denoting the \emph{rank deficiency} of $\tau$.
\end{theorem}
\begin{proof}
Let $\tau\in \mEnd$ be an endomorphism with rank deficiency $s\in \mathbb{N}_0$.
According to Theorem \ref{theorem:structure_of_end}, the mapping $\tau$ possesses a transformation matrix ${\bm{T}=\bm{A}\bm{Z}_T\bm{A}^{-1}}$ with 
$$\bm{Z}_T=\begin{bmatrix}  \bm{C}_T &\bm{0}_{(n-k)\times k} \\
\bm{D}_T& \bm{E}_{T}\end{bmatrix}.$$
Note that $\mathrm{Rank}(\bm{E}_{T})=k-s$.

According to the first isomorphism theorem, for every endomorphism $\tau:\mathcal{C}\to\mathcal{C}$, there exists an isomorphism between the quotient space $\mathcal{C}/\mathrm{Null}(\tau)$ and  $\mathrm{Im}(\tau)$. Therefore,
there exists an inverse mapping $\rho:\mathrm{Im}(\tau)\rightarrow\mathcal{C}/\mathrm{Null}(\tau)$. 

We are now interested in specifying a transformation matrix $\bm{R}$ of the mapping $\rho$.
To this end, we diagonalize $\bm{E}_T$ using Gaussian elimination in two steps yielding two matrices $\bm{G}_{\mathrm{r}},\bm{G}_{\mathrm{l}}$ required to determine the reconstruction matrix $\bm{R}$.
First, we perform Gaussian elimination on the columns with right-operating $\bm{G}_{\mathrm{r}}\in \mathrm{GL}_{k}(\mathbb{F}_q)$, transforming $\bm{E}_{T}$ into ${\bm{E}_{T}\bm{G}_{\mathrm{r}}\in \mathcal{M}}$, i.e., into an element of $\mathcal{M}$. This enables a simple definition of a basis of $\mathrm{Null}(\bm{E}_{T}\bm{G}_{\mathrm{r}})$.

Next, we diagonalize ${\bm{E}_{T}\bm{G}_{\mathrm{r}}}$ using row-wise Gaussian elimination represented by left multiplication with $\bm{G}_{\mathrm{l}}\!\in\! \mathrm{GL}_{k}(\mathbb{F}_q)$ such that according row operations yield
\begin{align*}
    \bm{\Lambda}&:= \bm{G}_{\mathrm{l}}\bm{E}_{T}\bm{G}_{\mathrm{r}}\\
    &=  \mathrm{diag}\left(\mathbbm{1}_{\left\{(\bm{E}_{T}\bm{G}_{\mathrm{r}})_{11}\neq 0\right\}},\ldots, \mathbbm{1}_{\left\{(\bm{E}_{T}\bm{G}_{\mathrm{r}})_{kk}\neq 0\right\}}\right),
\end{align*}
where $\mathbbm{1}_{\{x\}}$ is the indicator function.

Now, we show that $\bm{R}=\bm{A}\bm{Z}_{R}\bm{A}^{-1}$ with
\begin{equation}       
\bm{Z}_{R}=\begin{bmatrix}  \bm{C}_R &\bm{0}_{(n-k)\times k} \\
\bm{D}_R& \bm{E}_{R}\end{bmatrix}, \quad \bm{E}_{R}=\bm{G}_{\mathrm{r}}\bm{G}_{\mathrm{l}}
\end{equation}
fulfills (\ref{eq:reconstruction}), i.e., $\bm{R}$ is inverting $\bm{T}$ in $\mathrm{Im}(\tau)$.
To prove this statement, we derive suitable bases of both $\mathrm{Null}(\tau)$ and $\mathcal{C}$.

Let $\mathcal{J}:=\mathcal{J}(\bm{E}_{T}\bm{G}_{\mathrm{r}})$.
Since $\bm{E}_{T}\bm{G}_{\mathrm{r}}\in \mathcal{M}$, it follows that
the set ${\{\bm{e}_j\in\mathbb{F}_q^k\}_{j\in \mathcal{J}}}$ forms a basis of $\mathrm{Null}(\bm{E}_{T}\bm{G}_{\mathrm{r}})$.
Define 
\begin{alignat*}{2}
   \bm{\varepsilon}_{\mathrm{s},i}&:=\bm{G}_{\mathrm{r}}\bm{e}_i\in \mathbb{F}_q^k \qquad &&i\in \{1,\ldots,k\}\\
   \bm{\varepsilon}_i&:=\begin{pmatrix} 
\bm{0}_{n-k}\\
 \bm{\varepsilon}_{\mathrm{s},i}
\end{pmatrix}\in \mathbb{F}_q^n\qquad && i\in \{1,\ldots,k\}.
\end{alignat*}

Because $\bm{G}_{\mathrm{r}}$ is non-singular, the set $\left\{ \bm{\varepsilon}_{\mathrm{s},j}\right\}_{j\in \mathcal{J}}$
forms a basis of $\mathrm{Null}(\bm{E}_T)$
implying that $\left\{\bm{\varepsilon}_j\right\}_{j\in \mathcal{J}}$ forms a basis of ${\mathrm{Null}(\bm{Z}_T)\cap \left\{\bm{\varepsilon}_j\right\}_{j\in \mathcal{J}} }$.
Similar to the proof of \cite[Theorem 1]{mandelbaum2023generalized}, by using the fact that the CCM $\bm{A}$ is non-singular, it follows that $\left\{\bm{A}\bm{\varepsilon}_j\right\}_{j\in \mathcal{J}}$ forms a basis of $\mathrm{Null}(\tau)$ and
${\left\{\bm{A}\bm{\varepsilon}_i\right\}_{i=1}^k=\left\{\bm{A}\bm{\varepsilon}_j\right\}_{j\in \mathcal{J}} \cup  \left\{\bm{A}\bm{\varepsilon}_j\right\}_{j\in \mathcal{J}^{\mathsf{c}}} } $ forms a basis of $\mathcal{C}$.

By definition, it holds that $\mathcal{J}\cup \mathcal{J}^{\mathsf{c}}=\{1,\ldots,k\}$, hence, every codeword $\bm{x}\in \mathcal{C}$ can be expressed as
\begin{equation}\label{eq:decomposition}
    \bm{x}=\sum_{j\in \mathcal{J}^\mathsf{c}} \alpha_j\bm{A}\bm{\varepsilon}_j+\sum_{j\in \mathcal{J}}\alpha_j\bm{A}\bm{\varepsilon}_j, 
\end{equation}
with ${\alpha_j\in \mathbb{F}_q}$ for all ${j \in \{1,\ldots,k\}}$. Note that ${\sum_{j\in \mathcal{J}}\alpha_j\bm{A}\bm{\varepsilon}_j\in \mathrm{Null}(\tau)}$. It remains to show that $\bm{R}$ fulfils (\ref{eq:reconstruction}).
Using (\ref{eq:decomposition}), the linearity of matrix multiplication, $\tau(\bm{A}\varepsilon_j)=0$ for $j\in \mathcal{J}$, and plugging in the definitions of $\bm{T}$ and $\bm{R}$ it follows that
\begin{align}
    \bm{R}\bm{T}\bm{x}&= \bm{R}\left(\bm{T}\sum_{j\in \mathcal{J}^\mathsf{c}} \alpha_j\bm{A}\bm{\varepsilon}_j+\bm{T}\sum_{j\in \mathcal{J}}\alpha_j\bm{A}\bm{\varepsilon}_j\right)\nonumber\\
&=\sum_{j\in \mathcal{J}^\mathsf{c}} \alpha_j \bm{A}\begin{pmatrix} \bm{0}_{n-k} \\
\bm{G}_{\mathrm{r}}\bm{G}_{\mathrm{l}}\bm{E}_{T} \bm{\varepsilon}_{\mathrm{s},j}\end{pmatrix}.\label{eq:rtx}
\end{align}
We can further simply the lower part of the vector:
\begin{align}
    \bm{G}_{\mathrm{r}}\bm{G}_{\mathrm{l}}\bm{E}_{T} \bm{\varepsilon}_{\mathrm{s},j}&= \bm{G}_{\mathrm{r}}\bm{G}_{\mathrm{l}}\bm{E}_{T}\bm{G}_{\mathrm{r}}\bm{e}_j\nonumber\\
    &=\bm{G}_{\mathrm{r}}\bm{\Lambda}\bm{e}_j=\bm{G}_{\mathrm{r}}\bm{e}_j= \bm{\varepsilon}_{\mathrm{s},j}.\label{eq:back_to_basis}
\end{align}
where we used the fact that $\bm{\Lambda}\bm{e}_j=\bm{e}_j$ if $j\in \mathcal{J}^\mathsf{c}$.
Inserting (\ref{eq:back_to_basis}) into (\ref{eq:rtx}) results in
\[
    \bm{R}\bm{T}\bm{x}=\sum_{j\in \mathcal{J}^\mathsf{c}} \alpha_j \bm{A}\begin{pmatrix} \bm{0}_{n-k} \\
 \bm{\varepsilon}_{\mathrm{s},j}\end{pmatrix}=\sum_{j\in \mathcal{J}^\mathsf{c}} \alpha_j \bm{A}
\bm{\varepsilon}_j\in [\bm{x}]_{\tau}
\]
completing the proof.
\end{proof}
Note that given $\bm{R}$ and ${\mathrm{Span}\left(\left\{\bm{A}\bm{\varepsilon}_j:j\in \mathcal{J}\right\}\right)=\mathrm{Null}(\tau)}$,
we obtain all elements of the coset $[\bm{x}]_{\tau}$, i.e., all codewords ${\bm{x}+\sum_{j\in \mathcal{J}}\alpha_j \bm{A}\bm{\varepsilon}_j, \alpha_j\in \mathbb{F}_q}$ are mapped onto $\bm{T}\bm{x}$ by $\tau$. 
This fact will be exploited in Sec. \ref{sec:eed} for EED by mapping a path decoding estimate to a list of estimates that take part in the ML-in-the-list decision.

Next, we show that there exists a one-to-one mapping between the set of transformation matrices of endomorphisms $\mMatrixEnd$ and a code 
$\mathcal{C}\left(n^2,2kn-k^2\right)$. 
Afterwards, we demonstrate that knowledge of the automorphism group of a code can be exploited to find unknown endomorphisms.

\subsection{Transformation Matrices form a Linear Code}\label{subsection:end_as_code}
According to Theorem \ref{theorem:structure_of_end}, a mapping $\tau:\mathcal{C}\rightarrow \mathrm{Im}(\tau)\subseteq\mathcal{C}$ with transformation matrix $\bm{T}\in \mathbb{F}^{n\times n}$ is an endomorphism if and only if 
$\bm{A}^{-1}\bm{T}\bm{A}=\bm{Z}\in \mathcal{Z}$ regardless of the specific matrices $\bm{C},\bm{D},\bm{E}$ in (\ref{eq:set_z}). Hence, any $\bm{T}$ yielding the all-zero block in the upper right corner of $\bm{Z}$
is element of $\mMatrixEnd$. Next, we partition the CCMs according to
\begin{equation}\label{eq:partition_ccm}
\bm{A}^{-1}:=\begin{bmatrix}
    \bm{\Omega}_1\\
    \bm{\Omega}_2
\end{bmatrix},\quad \bm{A}:=\begin{bmatrix}
    \bm{A}_1&\bm{A}_2
\end{bmatrix},
\end{equation}
with $\bm{\Omega}_1\!\in\! \mathbb{F}_q^{(n-k)\times n}, \bm{\Omega}_2\! \in \!\mathbb{F}_q^{k\times n},    \bm{A}_1\!\in\! \mathbb{F}_{q}^{n\times k}, \bm{A}_2\!\in \!\mathbb{F}_q^{n\times (n-k)}$.
Then, using block matrix multiplications, we get
\[
\bm{Z}=\bm{A}^{-1}\bm{T}\bm{A}=\begin{bmatrix}
    \bm{\Omega}_1\bm{T}\bm{A}_1&\bm{\Omega}_1\bm{T}\bm{A}_2\\
    \bm{\Omega}_2\bm{T}\bm{A}_1&\bm{\Omega}_2\bm{T}\bm{A}_2
\end{bmatrix}.
\]
Hence, $\bm{T}$ is the transformation matrix of an endomorphism if and only if
\begin{equation}\label{eq:sufficient_condition}
    \bm{\Omega}_1\bm{T}\bm{A}_2=\bm{0}_{(n-k)\times k}.
\end{equation} Next, we show that there exists a
one-to-one mapping between the set $\mMatrixEnd$ and a code $\mathcal{C}\left(n^2,2kn-k^2\right)$. 
To this end, we use the $\mathrm{vec}$-operator defined in \cite{Petersen2006TheMC} and the Kronecker product~$\otimes$.
The one-to-one mapping $\mathrm{vec}:\mathbb{F}_q^{m_1\times m_2}\rightarrow\mathbb{F}_q^{m_1\cdot m_2}$ stacks the columns of a matrix to a column vector. Accordingly, ${\mathrm{vec}^{-1}
:\mathbb{F}_q^{m_1\cdot m_2}\rightarrow\mathbb{F}_q^{m_1\times m_2}}$ denotes the inverse mapping reshaping a column vector into the corresponding matrix.
\begin{proposition}\label{proposition:end_code}
    Consider a linear code $\mathcal{C}(n,k)$ with CCM $\bm{A}$ that is partitioned according to (\ref{eq:partition_ccm}).
    The PCM ${\bm{H}_{\mathrm{E}}=\bm{A}_2^{\mathsf{T}}\otimes\bm{\Omega}_1}$ defines a linear code $\mathcal{C}_{\mathrm{E}}\left(n^2,2kn-k^2\right)$.
    Furthermore, $ {\bm{x}\in \mathcal{C}_{\mathrm{E}}\left(n^2,2kn-k^2\right) \iff \bm{T}:=\mathrm{vec}^{-1}(\bm{x})\in \mathcal{T}_\mathrm{E}(\mathcal{C}(n,k)})$.
\end{proposition}
\begin{proof}
    Using the property $\mathrm{vec}(\bm{FJL})=(\bm{L}^{\mathsf{T}}\otimes\bm{F})\mathrm{vec}(\bm{J})$  {\cite[Eq. (520)]{Petersen2006TheMC}} and rearranging (\ref{eq:sufficient_condition}) yields
    \begin{align*}
         &\mathrm{vec}\left(  \bm{\Omega}_1\bm{T}\bm{A}_2\right)=\bm{0}_{k(n-k)}\\
         \iff &(\bm{A}_2^{\mathsf{T}}\otimes\bm{\Omega}_1)\mathrm{vec}(\bm{T})=\bm{0}_{k(n-k)}.
    \end{align*}
    Thus, the vectorized endomorphism matrices $\mathrm{vec}(\bm{T})$ are elements of the code $\mathcal{C}_{\mathrm{E}}$ defined by the PCM $\bm{H}_{\mathrm{E}}$. Furthermore, according to {\cite[Eq. (514)]{Petersen2006TheMC}} and because the CCM is non-singular, we obtain
    $$\mathrm{Rank}(\bm{H}_{\mathrm{E}})=\mathrm{Rank}\left(\bm{A}_2^{\mathsf{T}}\right)\cdot \mathrm{Rank}(\bm{\Omega}_1)=(n-k)^2$$
    implying that $\mathcal{C}_{\mathrm{E}}$ is a ${(2nk-k^2)}$-dimensional vector space.
\end{proof}
Proposition \ref{proposition:end_code} shows that 
finding code endomorphisms 
is equivalent to finding 
suitable codewords of a larger code. This in turn can be facilitated by integer programming \cite{8464884}.
To the best of our knowledge, there exists no such formulation for automorphisms. The rank constraint for automorphisms hinders a trivial equivalent formulation as an integer linear program.

\subsection{Endomorphism based on Automorphisms}
In literature, the automorphism groups of many codes are well-studied.
Thus, we consider the following construction of code endomorphisms based on its automorphisms. 
\begin{proposition}\label{proposition:superposition_auts}
     Consider a linear code $\mathcal{C}(n,k)$ with two automorphisms $\tau_1,\tau_2$ with transformation matrices $\bm{T}_1,\bm{T}_2$, respectively.
     Then, $\bm{T}:=\bm{T}_1+\bm{T}_2$ is the transformation matrix of an endormorphism $\tau\in \mEnd$.
\end{proposition}
\begin{proof}%
Consider an arbitrary codeword $\bm{x}\in\mathcal{C}$. Then, using the linearity of the code, it follows that
\[
\bm{T}\bm{x}=(\bm{T}_1+\bm{T}_2)\bm{x}=\bm{T}_1\bm{x}+\bm{T}_2\bm{x}\in \mathcal{C}.\hfill \qedhere
\]
\end{proof}

\section{Endomorphism Ensemble Decoding}\label{sec:eed}
\begin{figure}[t]
\centering
\begin{tikzpicture}	  

\node (input) at (-1.5,-0.5) {$\bm{y}$};    
\draw[fill] (-0.98,-0.5) circle (1pt);
\node (pi1) [rectangle, draw, thin, minimum width=1.4cm, minimum height=0.6cm] at (0.2,0.3) {Pre$(\cdot|\bm{T}_1$)};
\node (pik) [rectangle, draw, thin, minimum width=1.4cm, minimum height=0.6cm] at (0.2,-1.3) {Pre$(\cdot|\bm{T}_K$)};

\node (debox_1) [myblock, right=0.6cm of pi1] {Dec.};
\node (debox_K) [myblock, right=0.6cm of pik] {Dec.};

\node (pi_inv_1) [rectangle, draw, thin, minimum width=1.6cm, fill=black!20, minimum height=0.6cm,right=0.6cm of debox_1]  {Post$(\cdot|\bm{T}_1$)};
\node (pi_inv_k) [rectangle, draw, thin, minimum width=1.6cm, fill=black!20, minimum height=0.6cm,right=0.6cm of debox_K]  {Post$(\cdot|\bm{T}_K$)};

\node (ml_1) [right=0.7cm of pi_inv_1] {};
\node (ml_K) [right=0.7cm of pi_inv_k] {};

\node (ml_box_inv1) [rectangle, minimum width=0.4cm, minimum height=0.6cm,right=1.1cm of pi_inv_1]  {};
\node (ml_box_invk) [rectangle, minimum width=0.4cm, minimum height=0.6cm,right=1.1cm of pi_inv_k]  {};
\node (ml_box) [rectangle, draw, thin, minimum width=0.6cm, minimum height=2.2cm] at (5.9,-0.5)  {\rotatebox{90}{ML-in-the-list}};
\node (output) [right=0.3cm of ml_box] {$\hat{\bm{x}}$};

\node (dots_box) [below=0cm of pi1]  {$\vdots$};
\node (dots_box) [below=0cm of debox_1]  {$\vdots$};
\node (dots_pi_inv_box) [below=0cm of pi_inv_1] {$\vdots$};
\myline{input}{pi1};
\myline{input}{pik};

\draw [-latex] (pi1) -- (debox_1) node[draw=none,fill=none,font=\scriptsize,midway,above] {$\bm{L}_{\tau_1}$};
\draw [-latex] (pik) -- (debox_K) node[draw=none,fill=none,font=\scriptsize,midway,above] {$\bm{L}_{\tau_K}$};

\draw [-latex] (debox_1) -- (pi_inv_1) node[draw=none,fill=none,font=\scriptsize,midway,above] {$\hat{\bm{x}}_{\tau_1}$};
\draw [-latex] (debox_K) -- (pi_inv_k) node[draw=none,fill=none,font=\scriptsize,midway,above] {$\hat{\bm{x}}_{\tau_K}$};

\draw [-latex] (ml_box) -- (output);	

\foreach \x in {.1,.3,.9} {
\draw[-latex] ($(pi_inv_1.north east)!\x!(pi_inv_1.south east)$) -- ($(ml_box_inv1.north west)!\x!(ml_box_inv1.south west)$);
}
\foreach \x in {.1,.3,.9} {
\draw[-latex] ($(pi_inv_k.north east)!\x!(pi_inv_k.south east)$) -- ($(ml_box_invk.north west)!\x!(ml_box_invk.south west)$);
}
\node (coset_1) [above right=-0.1cm and -0.05cm of pi_inv_1]  {$[\widehat{\bm{x}}_1]_{\tau_1}$};
\node (coset_k) [above right=-0.1cm and -0.05cm of pi_inv_k]  {$[\widehat{\bm{x}}_K]_{\tau_K}$};

\node (dots_box_1) [right=0.2cm of pi_inv_1]  {\scalebox{0.8}{\tiny{\vdots}}};
\node (dots_box_k) [right=0.2cm of pi_inv_k]  {\scalebox{0.8}{\tiny{\vdots}}};

\end{tikzpicture}
\caption{Block diagram of EED using $K$ different endomorphisms with transformation matrices ${\bm{T}_i \in \mMatrixEnd}$.}\label{figure:eed}
\vspace{-0.5cm}
\end{figure}
It is well known that automorphisms of codes can be used to improve decoding, e.g., with AED \cite{AED_RMcodes}. 
Thus, a possible application of the analysis of endomorphisms is ensemble decoding. Therefore, in this section, we propose EED, a decoding scheme using endomorphisms in an ensemble decoder thereby extending AED \cite{AED_RMcodes} and GAED  \cite{mandelbaum2023generalized}.

To this end, we consider the transmission of a binary codeword ${\bm{x}\in \mathcal{C}\subset \mathbb{F}_2^n}$ over a binary memoryless symmetric channel \cite[Ch. 4]{MCT08}.
Note that we constrain ourselves to binary codes. Yet, the scheme can easily be extended to linear block codes over arbitrary finite fields.
The receiver observes ${\bm{y}\in\mathcal{Y}^n}$ with $\mathcal{Y}$ being the channel output alphabet, or, equivalently, the LLR vector
${\bm{L}:=\left(L(y_j|X_j)\right)_{j=1}^n\in \mathbb{R}^n}$, with uppercase letters denoting random variables.
Furthermore, let ${\bm{T}_1,\ldots,\bm{T}_K \in \mMatrixEnd}$ be the transformation matrices of $K$ arbitrary endomorphisms. %
Then, EED, as depicted in Fig.~\ref{figure:eed}, consists of $K$ parallel paths, one per chosen endomorphism, each incorporating a preprocessing block, a decoder block, and a post-processing block.

Consider now an arbitrary path with mapping $\bm{T}$. 
Then, to mimic the effect of the summation in $\mathbb{F}_2$ conducted in (\ref{eq:linear_mapping}) in the LLR domain,
we process the LLR vector $\bm{L}$ with the preprocessing described in \cite{mandelbaum2023generalized}:
\begin{equation}\label{eq:boxplus}
\left(\bm{L}_{\tau}\right)_j\!\!:=\left(\text{Pre}(\bm{y}|\bm{T})\right)_j:= \!\!\!\!\mathboxplus_{
\begin{subarray}{l}
\,\,\,\, i=1,\\
T_{j,i}=1 \end{subarray}}^{n}  L(y_i|X_i),\,\, \forall j \in \{1,\ldots,n\}
\end{equation}
using the box-plus operator $\mathboxplus$ introduced in \cite{485714}.

Afterwards, the preprocessed LLR vector is decoded using the respective path decoder of the original code $\mathcal{C}$ yielding an estimate $\hat{\bm{x}}_\tau$ of the transmitted codeword after applying the respective linear mapping $\bm{T}$.
Endomorphisms are in general not injective and multiple transmitted codewords may yield the same estimate.

Therefore, we propose a post-processing step based on Theorem~\ref{theorem:reconstruction_end},
mapping an estimate $\hat{\bm{x}}_\tau\in \mathrm{Im}(\tau)$ onto the coset
$[\bm{R}\hat{\bm{x}}_\tau]_\tau=\left\{\bm{R}\hat{\bm{x}}_\tau+\sum_{j\in\mathcal{J}}\alpha_j\bm{A\varepsilon}_j,\alpha_j\in \mathbb{F}_2\right\}$
with cardinality $2^s$, i.e., 
a list of codeword candidates mapped onto $\hat{\bm{x}}_\tau$ by $\tau$.
If the decoding path fails to converge to an element of $\mathrm{Im}(\tau)$, its output is discarded and the path contributes no estimate.

Last, the estimates of all paths are collected in an \emph{ML-in-the-list} block and the final estimate $\hat{\bm{x}}$ is obtained according to the ML-in-the-list rule \cite{AED_RMcodes}.

It is noteworthy that \cite{mandelbaum2023generalized} discusses the influence of the \emph{weight over permutation} $\Delta\left(\bm{T}\right)$
of a mapping $\bm{T}$, defined as ${\Delta\left(\bm{T}\right)=\sum_{i,j}T_{i,j} - n}$, on the preprocessing. It is shown that an increasing number of non-zero entries compared to a transformation matrix of a permutation leads to diminishing reliability during the preprocessing.
Therefore, mappings with small weights over permutation are required to improve the ensemble decoding performance.

Note that, if only generalized automorphisms are used, EED reverts to GAED. 
Hence, EED naturally generalizes GAED by introducing a new post-processing step 
such that arbitrary endomorphisms can be used in an ensemble decoding scheme.

The reconstruction matrix $\bm{R}$ and the required basis of $\mathrm{Null}(\tau)$ for the post-processing block can be pre-computed at design time.
Therefore, the same holds for all vectors in the set $\left\{\sum_{j\in\mathcal{J}}\alpha_j\bm{A\varepsilon}_j,\alpha_j\in \mathbb{F}_q\}\right\}$.
Assume an EED with $K$ endomorphisms with rank deficiencies $s_1,\ldots,s_K$.
Then, compared to GAED, the $i$-th path of EED requires $q^{s_i}-1$ additional $n$-dimensional vector additions.
These operations can be fully parallelized.
Hence, for small $s=\max\{s_1,\ldots,s_K\}$, the complexity and latency are comparable to GAED and, hence, to AED.
It can be stated that EED provides a general and flexible framework of ensemble decoding schemes incorporating GAED, AED, as well as MBBP.

\section{Numerical Results}
In this section, we demonstrate that EED can lower the frame error rate (FER) compared to belief propagation (BP) and successive-cancellation (SC) decoding for a Golay and a short-length polar code, respectively.
For more details on BP and SC decoding, the reader is directed to \cite{MCT08} and \cite{polar:arikan09}, respectively.

We conduct Monte-Carlo simulations for the extended Golay code $\mathcal{C}_{\mathrm{G}}(24,12)$
and a short-length polar code $\mathcal{C}_{\mathrm{P}}(32,16)$ used in the $5$G-NR standard over a binary-input AWGN channel \cite{8962344}.
We collect a minimum of $200$ frame errors for each data point.
The notation $\mathrm{EED}$-$\ell$-$\{\mathrm{BP/SC}\}$ refers to EED comprising ${\ell\text{-paths}}$ each using BP or SC decoding, respectively.
All EEDs employ the identity mapping in their respective $1$-th path which typically possesses the best stand-alone decoding performance. Hence, the additional paths are called auxiliary paths.

\begin{figure}
    \centering
    \begin{tikzpicture}[scale=0.92,spy using outlines={rectangle, magnification=2}]

\begin{axis}[%
width=.85\columnwidth,
height=4.3cm,
at={(0.758in,0.645in)},
scale only axis,
xmin=1,
xmax=5,
xlabel style={font=\color{white!15!black}},
xlabel={$E_{\mathrm{b}}/N_0$ ($\si{dB}$)},
ymode=log,
ymin=1e-04,
ymax=1,
yminorticks=true,
ylabel style={font=\color{white!15!black}},
ylabel={FER},
axis background/.style={fill=white},
xmajorgrids,
ymajorgrids,
legend style={at={(0.03,0.03)}, anchor=south west, legend cell align=left, align=left, draw=white!15!black,font=\scriptsize}
]

\addplot [color=black, line width=0.9pt]
  table[row sep=crcr]{%
0.00   2.594e-01\\  %
0.50   1.986e-01\\   %
1.00   1.486e-01\\   %
1.50   8.482e-02\\   %
2.00   4.890e-02\\   %
2.50   2.308e-02\\   %
3.00   1.114e-02\\   %
3.50   4.369e-03\\   %
4.00   1.869e-03\\   %
4.50   5.622e-04\\   %
5.00   1.477e-04\\   %
5.50   3.287e-05\\   %
};
\addlegendentry{ML \cite{channelcodes}}%

\addplot [color=KITblue, line width=1.1pt,mark=triangle]
  table[row sep=crcr]{%
0.000000000000000000e+00 6.416666666666667185e-01\\
5.000000000000000000e-01 4.787644787644787514e-01\\
1.000000000000000000e+00 3.917963224893917795e-01\\
1.500000000000000000e+00 2.901498929336188692e-01\\
2.000000000000000000e+00 2.160727824109173745e-01\\
2.500000000000000000e+00 1.426553672316384080e-01\\
3.000000000000000000e+00 8.112505811250581012e-02\\
3.500000000000000000e+00 4.924898902368573389e-02\\
4.000000000000000000e+00 2.621559203545346592e-02\\
4.500000000000000000e+00 1.288167938931297773e-02\\
5.000000000000000000e+00 5.251906514008509996e-03\\
};
\addlegendentry{BP}

\addplot [color=KITorange, line width=1.1pt, mark=+]
  table[row sep=crcr]{%
0.0 5.084746e-01\\
0.5 4.103967e-01\\
1.0 3.322259e-01\\
1.5 2.286585e-01\\
2.0 1.397950e-01\\
2.5 8.828723e-02\\
3.0 4.705144e-02\\
3.5 2.271695e-02\\
4.0 1.057418e-02\\
4.5 4.017301e-03\\
5.0 1.205531e-03\\
};
\addlegendentry{$4$-MBBP \cite{MBBP2}}

\addplot [color=KITgreen, line width=1.1pt, mark=asterisk]
  table[row sep=crcr]{%
0.000000000000000000e+00 4.032258064516128782e-01\\
5.000000000000000000e-01 2.886002886002885792e-01\\
1.000000000000000000e+00 2.424242424242424310e-01\\
1.500000000000000000e+00 1.404494382022472010e-01\\
2.000000000000000000e+00 9.289363678588016815e-02\\
2.500000000000000000e+00 6.295247088448222006e-02\\
3.000000000000000000e+00 3.406574689150059648e-02\\
3.500000000000000000e+00 1.625884074465490728e-02\\
4.000000000000000000e+00 7.786645902277593810e-03\\
4.500000000000000000e+00 3.406052555390929552e-03\\
5.000000000000000000e+00 1.287307788855776396e-03\\
};
\addlegendentry{EED-4\--BP} %

\end{axis}
\end{tikzpicture}%
    \caption{ Performance of different decoders for the extended Golay code $\mathcal{C}_{\mathrm{G}}(24,12)$. All BP decodings are based on an overcomplete PCM following the construction proposed in \cite{MBBP2}. The auxiliary paths of EED use endomorphisms with weight over permutation ${\Delta(\bm{T})=8}$ and rank deficiency $s=8$.}
    \label{fig:golay_res}
    \begin{tikzpicture}[scale=0.92,spy using outlines={rectangle, magnification=2}]

\begin{axis}[%
width=.85\columnwidth,
height=4.3cm,
at={(0.758in,0.645in)},
scale only axis,
xmin=1,
xmax=5,
xlabel style={font=\color{white!15!black}},
xlabel={$E_{\mathrm{b}}/N_0$ ($\si{dB}$)},
ymode=log,
ymin=1e-04,
ymax=1,
yminorticks=true,
ylabel style={font=\color{white!15!black}},
ylabel={FER},
axis background/.style={fill=white},
xmajorgrids,
ymajorgrids,
legend style={at={(0.03,0.03)}, anchor=south west, legend cell align=left, align=left, draw=white!15!black,font=\scriptsize}
]

\addplot [color=black, line width=0.9pt]
  table[row sep=crcr]{%
 1.0  1.547988e-01\\
 1.5  1.027485e-01\\
 2.0  5.831754e-02\\
 2.5  2.998501e-02\\
 3.0  1.409741e-02\\
 3.5  5.323042e-03\\
 4.0  1.826768e-03\\
 4.5  5.168486e-04\\
 5.0  1.366593e-04\\
};
\addlegendentry{ML (OSD-3)}%

\addplot [color=KITblue, line width=1.1pt,mark=triangle]
  table[row sep=crcr]{%
1.000000000000000000e+00 2.379327157484154021e-01\\
1.500000000000000000e+00 1.742053789731051461e-01\\
2.000000000000000000e+00 1.268576351496877030e-01\\
2.500000000000000000e+00 7.298764291488624156e-02\\
3.000000000000000000e+00 4.024375527935320634e-02\\
3.500000000000000000e+00 1.973258851216799781e-02\\
4.000000000000000000e+00 8.923494160360439034e-03\\
4.500000000000000000e+00 3.497849470325651612e-03\\
5.000000000000000000e+00 1.137538093140910787e-03\\
6.000000000000000000e+00 7.299999999999999903e-05\\
};
\label{plot:polar_sc}
\addlegendentry{SC = AED-SC (LTA)};

\addplot [color=KITgreen, line width=1.1pt, mark=asterisk]
  table[row sep=crcr]{%
1.000000000000000000e+00 1.950268161872257544e-01\\
1.500000000000000000e+00 1.222493887530562290e-01\\
2.000000000000000000e+00 8.615119534783545474e-02\\
2.500000000000000000e+00 4.619471070562420484e-02\\
3.000000000000000000e+00 2.413418607457463658e-02\\
3.500000000000000000e+00 1.034473840742752143e-02\\
4.000000000000000000e+00 4.374261843313941196e-03\\
4.500000000000000000e+00 1.480571204370646245e-03\\
5.000000000000000000e+00 4.426218261248680342e-04\\
6.000000000000000000e+00 1.900000000000000105e-05\\
};
\label{plot:polar_autsc}
\addlegendentry{EED-$4$-SC};

\end{axis}
\end{tikzpicture}%
    \caption{ Performance of different decoders for the 5G polar code $\mathcal{C}_{\mathrm{P}}(32,16)$. Note that no CRC is used for the polar code. The auxiliary paths of EED use endomorphisms with weight over permutation ${\Delta(\bm{T})=16}$ and rank deficiency $s=8$.}
    \label{fig:polar_res}
\end{figure}

Due to their relation to Steiner systems, Golay codes possess a large automorphism group known as the Mathieu group \cite{MacWilliamsSloane}. Furthermore, in \cite{algebraicpropertiesofpolarcodes}, the authors show that the automorphism group of polar codes contains the lower triangular affine group (LTA). Depending on the choice of frozen positions, a possibly larger automorphism group is known \cite{Aut_PolarCodes_Geiselhart}.\footnote{The largest automorphism group is obtained if the frozen indices are chosen such that an RM code is obtained \cite{MacWilliamsSloane}.}

By sampling pairs of automorphisms from the Mathieu group $\mathrm{M}_{24}$ for the Golay code $\mathcal{C}_{\mathrm{G}}(24,12)$ and from the LTA for the polar code $\mathcal{C}_{\mathrm{P}}(32,16)$\footnote{Note that $\mathcal{C}_{\mathrm{P}}(32,16)$ is an RM code. Hence, its automorphism group is the larger general affine group. Yet, we constrain ourselves to elements from LTA to showcase a scenario in which AED yields no gains over SC decoding.}, respectively, and superposing their transformation matrices, 
we obtain proper endomorphisms with varying weight over permutation and rank deficiency.
Among those, we arbitrarily choose 3 endomorphisms with rank deficiency $8$ and weight over permutations $8$ for the Golay and $16$ for the polar code, respectively,  to obtain a small weight over permutation and moderate rank deficiency.

All BP decoders perform $32$ iterations of the sum-product algorithm
with a flooding schedule on the Tanner graph defined by an overcomplete PCM constructed as in \cite{MBBP2}. For SC decoding, we use the implementation from \cite{6065237}.

Figures~\ref{fig:golay_res} and \ref{fig:polar_res} depict the FER over $E_{\mathrm{b}}/N_0$ of EED-$4$-BP and EED-$4$-SC compared to stand-alone BP and SC decoding, respectively. Furthermore, we demonstrate the performance of respective ML decoding. In Fig.~\ref{fig:golay_res}, we additionally show the performance of MBBP decoding according to \cite{MBBP2} using an equal number of paths and decoding iterations.

The performance of ${\text{EED-}4\text{-BP}}$ and ${\text{EED-}4\text{-SC}}$ improves over the full simulated SNR regime compared to BP decoding and SC decoding, respectively.
In particular, at an FER of $10^{-2}$, we observe a gain of $0.8\,\si{dB}$ and $0.4\,\si{dB}$, respectively,
demonstrating that both BP and SC decoding can be significantly improved by EED.

At low SNR regime, EED-$4$-BP possesses a lower FER compared to MBBP. At an FER of $10^{-2}$, we observe a gain of $0.2\,\si{dB}$. However, this is at the expense of increased complexity due to the post-processing step.
With increasing SNR, the gain diminishes until both decoders yield a FER of $1.2 \cdot 10^{-3}$.

It is important to highlight that SC decoding can not be improved using an AED based on automorphisms from the LTA \cite[Corollarly 2.1]{AED_RMcodes}. Interestingly, using endomorphisms constructed based on these automorphisms can improve SC decoding, demonstrating the increased flexibility of EED.

\section{Conclusion}
In this work, we considered endomorphisms of linear block codes proving different structural properties enabling 
 an explicit construction of their transformation matrices. 
Furthermore,
we derived a one-to-one mapping between the set of transformation matrices of endomorphisms and a larger linear block code
enabling the usage of well-known algorithms to find endomorphisms for different applications, e.g., for decoding.
Then, we proposed a possible application of proper endomorphisms in ensemble decoding introducing EED showing improved performance both compared to BP and SC decoding for some short codes. In particular, we highlighted a scenario in which knowledge of the automorphism can only be exploited in EED but not in AED.
To conclude, endomorphisms of codes promise to provide additional structural insights into the code and to improve decoding, especially in the short-block length regime.

\newpage

\appendix
\textbf{Example} To illustrate the introduced definitions and the post-processing step of EED we provide a detailed example in which we consider a $\mathcal{C}_{\mathrm{H}}(7,4)$ Hamming code defined by the PCM
\[\bm{H}=\left(\begin{smallmatrix}
  1 & 0 & 1 & 1 & 1 & 0 & 0 \\
  1 & 1 & 0 & 1 & 0 & 1 & 0 \\
  0 & 1 & 1 & 1 & 0 & 0 & 1 \\
\end{smallmatrix}\right).
\]
First, we perform Gaussian elimination on the columns of the PCM $\bm{H}$ determining a CCM and its inverse
\[
\bm{A}=\left(
\begin{smallmatrix}
  0 & 1 & 1 & 1 & 0 & 1 & 1 \\
  1 & 1 & 0 & 1 & 1 & 1 & 0 \\
  0 & 0 & 0 & 1 & 0 & 0 & 0 \\
  1 & 1 & 1 & 0 & 1 & 1 & 1 \\
  0 & 0 & 0 & 0 & 1 & 0 & 0 \\
  0 & 0 & 0 & 0 & 0 & 1 & 0 \\
  0 & 0 & 0 & 0 & 0 & 0 & 1 \\
\end{smallmatrix}\right), \qquad \bm{A}^{-1}=
\left(\begin{smallmatrix}
   1 & 0 & 1 & 1 & 1 & 0 & 0 \\
  1 & 1 & 0 & 1 & 0 & 1 & 0 \\
  0 & 1 & 1 & 1 & 0 & 0 & 1 \\
  0 & 0 & 1 & 0 & 0 & 0 & 0 \\
  0 & 0 & 0 & 0 & 1 & 0 & 0 \\
  0 & 0 & 0 & 0 & 0 & 1 & 0 \\
  0 & 0 & 0 & 0 & 0 & 0 & 1 \\
\end{smallmatrix}\right),
\]
such that $\bm{HA}=\left[\bm{I}_{3\times3}\,\bm{0}_{3\times 4} \right]$.
Next, we arbitrarily sample a matrix from $\mathcal{Z}$
\[
\bm{Z}_T=\left(\begin{smallmatrix}
 0 & 0 & 1 & 0 & 0 & 0 & 0 \\
  0 & 1 & 0 & 0 & 0 & 0 & 0 \\
  1 & 0 & 1 & 0 & 0 & 0 & 0 \\
  0 & 0 & 0 & 0 & 0 & 0 & 0 \\
  0 & 0 & 0 & 1 & 0 & 1 & 0 \\
  0 & 0 & 0 & 0 & 0 & 0 & 1 \\
  0 & 0 & 0 & 0 & 1 & 0 & 0 \\
\end{smallmatrix}\right)\in \mathcal{Z},
\]
with $\bm{D}_T=\bm{0}_{4\times3}$ and
\[
\bm{C}_T=\left(
\begin{smallmatrix}
 0 & 0 & 1 \\
  0 & 1 & 0 \\
  1 & 0 & 1 \\
\end{smallmatrix}\right), \qquad \bm{E}_T=
\left(\begin{smallmatrix}
  0 & 0 & 0 & 0 \\
  1 & 0 & 1 & 0 \\
  0 & 0 & 0 & 1 \\
  0 & 1 & 0 & 0 \\
\end{smallmatrix}\right).
\]
Note that $\mathrm{Rank}(\bm{E})=3$.
According to Theorem~\ref{theorem:structure_of_end}, we obtain the transformation matrix $\bm{T}$ of a (proper) endomorphism $\tau$ of the Hamming code
\[
\bm{T}=\bm{A}\bm{Z}_T\bm{A}^{-1}=\left(\begin{smallmatrix}
 0 & 0 & 0 & 1 & 0 & 1 & 0 \\
  1 & 0 & 0 & 0 & 0 & 0 & 0 \\
  0 & 0 & 0 & 0 & 0 & 0 & 0 \\
  0 & 1 & 0 & 0 & 0 & 0 & 1 \\
  0 & 0 & 1 & 0 & 0 & 1 & 0 \\
  0 & 0 & 0 & 0 & 0 & 0 & 1 \\
  0 & 0 & 0 & 0 & 1 & 0 & 0 \\
\end{smallmatrix}\right).
\]
We stack all $2^4$ codewords of $\mathcal{C}_{\mathrm{H}}$ as columns to a matrix 
\[\bm{W}=
\left(\begin{smallmatrix}
  0 & 0 & 0 & 0 & 0 & 0 & 0 & 0 & 1 & 1 & 1 & 1 & 1 & 1 & 1 & 1 \\
  0 & 0 & 0 & 0 & 1 & 1 & 1 & 1 & 0 & 0 & 0 & 0 & 1 & 1 & 1 & 1 \\
  0 & 1 & 0 & 1 & 1 & 0 & 0 & 1 & 1 & 0 & 1 & 0 & 1 & 0 & 0 & 1 \\
  0 & 1 & 1 & 0 & 1 & 0 & 1 & 0 & 1 & 0 & 0 & 1 & 0 & 1 & 0 & 1 \\
  0 & 0 & 1 & 1 & 0 & 0 & 1 & 1 & 1 & 1 & 0 & 0 & 0 & 0 & 1 & 1 \\
  0 & 1 & 1 & 0 & 0 & 1 & 0 & 1 & 0 & 1 & 1 & 0 & 0 & 1 & 0 & 1 \\
  0 & 0 & 1 & 1 & 1 & 1 & 0 & 0 & 0 & 0 & 1 & 1 & 0 & 0 & 1 & 1 \\
\end{smallmatrix}\right)
\]
and multiply it from the left with $\bm{T}$ yielding
\[
\bm{T}\bm{W}=\left(\begin{smallmatrix}
  0 & 0 & 0 & 0 & 1 & 1 & 1 & 1 & 1 & 1 & 1 & 1 & 0 & 0 & 0 & 0 \\
  0 & 0 & 0 & 0 & 0 & 0 & 0 & 0 & 1 & 1 & 1 & 1 & 1 & 1 & 1 & 1 \\
  0 & 0 & 0 & 0 & 0 & 0 & 0 & 0 & 0 & 0 & 0 & 0 & 0 & 0 & 0 & 0 \\
  0 & 0 & 1 & 1 & 0 & 0 & 1 & 1 & 0 & 0 & 1 & 1 & 1 & 1 & 0 & 0 \\
  0 & 0 & 1 & 1 & 1 & 1 & 0 & 0 & 1 & 1 & 0 & 0 & 1 & 1 & 0 & 0 \\
  0 & 0 & 1 & 1 & 1 & 1 & 0 & 0 & 0 & 0 & 1 & 1 & 0 & 0 & 1 & 1 \\
  0 & 0 & 1 & 1 & 0 & 0 & 1 & 1 & 1 & 1 & 0 & 0 & 0 & 0 & 1 & 1 \\
\end{smallmatrix}\right).
\]
We observe, that always two columns of $\bm{TW}$ are equal, i.e., two codewords from the Hamming code $\mathcal{C}_{\mathrm{H}}$ are mapped onto the same codeword.
Note that we conveniently ordered $\bm{W}$ such that the pair of neighbouring columns $(\bm{W}_{:,(2i-1)},\bm{W}_{:,2i})$ with $i\in\{1,\ldots,8\}$ are mapped onto the same codeword. Hence, those pairs form the cosets $[\bm{x}]_\tau$ with cardinality $2^1=2$, i.e., the rank deficiency is $s=1$. The rank deficiency can also be calculated according to $s=k-{\mathrm{Rank}(\bm{E}_T)}$.

We continue by determining step by step the reconstruction matrix $\bm{R}$ as well as the basis of $\mathrm{Im}(\tau)$.
To this end, we follow the proof of Theorem \ref{theorem:reconstruction_end}. We start by diagonalizing $\bm{E}$ in two steps. First, we perform Gaussian elimination on the columns represented by right-operating $\bm{G}_{\mathrm{r}}$ such that $\bm{EG}_{\mathrm{r}}\in \mathcal{M}$ and, second, 
we use row-wise Gaussian elimination represented by left-operating $\bm{G}_{\mathrm{l}}$ such that
$ {\bm{G}_{\mathrm{l}}\bm{EG}_{\mathrm{r}}=\mathrm{diag}\left(\mathbbm{1}_{\left\{(\bm{E}_{T}\bm{G}_{\mathrm{r}})_{11} \neq 0 \right\}},\ldots, \mathbbm{1}_{\left\{(\bm{E}_{T}\bm{G}_{\mathrm{r}})_{44}\neq 0\right\}}\right)}$
 yielding
\[
\bm{G}_{\mathrm{r}}=\left(\begin{smallmatrix}
    1 & 1 & 0 & 0 \\
  0 & 0 & 0 & 1 \\
  1 & 0 & 0 & 0 \\
  0 & 0 & 1 & 0 \\ 
\end{smallmatrix}\right)  ,\qquad \bm{G}_{\mathrm{l}}=\left(\begin{smallmatrix}
  1 & 0 & 0 & 0 \\
  0 & 1 & 0 & 0 \\
  0 & 0 & 1 & 0 \\
  0 & 0 & 0 & 1 \\
\end{smallmatrix}\right).
\]
It can be checked that
\begin{align*} 
\bm{E}_T\bm{G}_{\mathrm{r}}&=\left(\begin{smallmatrix}
  0 & 0 & 0 & 0 \\
  0 & 1 & 0 & 0 \\
  0 & 0 & 1 & 0 \\
  0 & 0 & 0 & 1 \\
\end{smallmatrix}\right)\in\mathcal{M},\\
\bm{\Lambda}&=\mathrm{diag}\left(\mathbbm{1}_{\left\{(\bm{E}_{T}\bm{G}_{\mathrm{r}})_{11}\neq 0\right\}},\ldots, \mathbbm{1}_{\left\{(\bm{E}_{T}\bm{G}_{\mathrm{r}})_{44}\neq 0\right\}}\right)\\
&=\left(\begin{smallmatrix}
  0 & 0 & 0 & 0 \\
  0 & 1 & 0 & 0 \\
  0 & 0 & 1 & 0 \\
  0 & 0 & 0 & 1 \\
\end{smallmatrix}\right).
\end{align*}
Hence, we obtain $\mathcal{J}=\{1\}$. Now, using
$$\bm{\varepsilon}_i=\begin{pmatrix} 
\bm{0}_{3\times1}\\
 \bm{G}_{\mathrm{r}}\bm{e}_i
\end{pmatrix}$$ we determine the sets
\begin{align*}
\left\{\bm{A}\bm{\varepsilon}_j\right\}_{j\in \mathcal{J}}&=\left\{\left(\begin{smallmatrix}
0\\
0\\
1\\
1\\
0\\
1\\
0\\
\end{smallmatrix}\right)\right\},\\
\left\{\bm{A}\bm{\varepsilon}_i\right\}_{i=1}^4&=\left\{\bm{A}\bm{\varepsilon}_j\right\}_{j\in \mathcal{J}}\cup \left\{\bm{A}\bm{\varepsilon}_j\right\}_{j\in \mathcal{J}^{\mathsf{c}}} \\
&=\left\{
\left(\begin{smallmatrix}
0\\
0\\
1\\
1\\
0\\
1\\
0\\
\end{smallmatrix}\right),
\left(
\begin{smallmatrix}
1\\
1\\
1\\
0\\
0\\
0\\
0\\
\end{smallmatrix}\right),
\left(\begin{smallmatrix}
1\\
0\\
0\\
1\\
0\\
0\\
1\\
\end{smallmatrix}\right),
\left(\begin{smallmatrix}
0\\
1\\
0\\
1\\
1\\
0\\
0\\
\end{smallmatrix}\right)\right\},
\end{align*}
consisting of linearly independent vectors, respectively. First, note that all elements of $\left\{\bm{A}\bm{\varepsilon}_i\right\}_{i=1}^4$ are codewords. Hence, $\left\{\bm{A}\bm{\varepsilon}_i\right\}_{i=1}^4$ forms a basis of $\mathcal{C}_{\mathrm{H}}(7,4)$.

Furthermore, note that 
\[
\bm{T}\left(\begin{smallmatrix}
0\\
0\\
1\\
1\\
0\\
1\\
0\\
\end{smallmatrix}\right)=\bm{0}_{7\times 1}.
\]
Hence, together with the rank deficiency $s=1$, the set $\left\{\bm{A}\bm{\varepsilon}_j\right\}_{j\in \mathcal{J}}$ indeed forms a basis of $\mathrm{Null}(\tau)$.
Finally, we calculate the reconstruction matrix 
\begin{align*}
\bm{R}&=\bm{A}\left(\begin{smallmatrix}
    \bm{0}_{3\times 3}&\bm{0}_{3\times 4}\\
    \bm{0}_{4\times 3}&\bm{G}_{\mathrm{r}}\bm{G}_{\mathrm{l}}\\
\end{smallmatrix}\right)\bm{A}^{-1}\\
&=\left(\begin{smallmatrix}
  0 & 0 & 0 & 0 & 1 & 1 & 0 \\
  0 & 0 & 0 & 0 & 1 & 0 & 1 \\
  0 & 0 & 1 & 0 & 1 & 0 & 0 \\
  0 & 0 & 1 & 0 & 0 & 1 & 1 \\
  0 & 0 & 0 & 0 & 0 & 0 & 1 \\
  0 & 0 & 1 & 0 & 0 & 0 & 0 \\
  0 & 0 & 0 & 0 & 0 & 1 & 0 \\
\end{smallmatrix}\right)
\end{align*}
assuming that $\bm{C}_R$ and $\bm{D}_R$ are all-zero matrices. Note that arbitrary $\bm{C}_R$ and $\bm{D}_R$ yield a valid reconstruction matrix.

Last, to illustrate the mapping with $\bm{T}$ and the post-processing step according to Theorem~\ref{theorem:reconstruction_end}, we consider an arbitrary codeword, e.g., ${\bm{x}=\left(\begin{smallmatrix}
  1 & 1 & 1 & 1 & 1 & 1 & 1
\end{smallmatrix}\right)^{\mathsf{T}}}$.
We obtain
\[
\bm{x}_\tau=\bm{Tx}=\left(\begin{smallmatrix}
0\\
1\\
0\\
0\\
0\\
1\\
1\\
\end{smallmatrix}\right)\in \mathrm{Im}(\tau)\subseteq\mathcal{C}_{\mathrm{H}}.
\]
Next, we determine the coset consisting of codewords mapped onto $\bm{x}_\tau$, i.e.,
\begin{align*}
\left\{\bm{R}\bm{x}_\tau+\alpha \left(\begin{smallmatrix}
0\\
0\\
1\\
1\\
0\\
1\\
0\\
\end{smallmatrix}\right): \alpha\in \{0,1\}\right\}&=\left\{\left(\begin{smallmatrix}
1\\
1\\
0\\
0\\
1\\
0\\
1\\
\end{smallmatrix}\right),\left(\begin{smallmatrix}
1\\
1\\
0\\
0\\
1\\
0\\
1\\
\end{smallmatrix}\right)\!\!+\!\! \left(\begin{smallmatrix}
0\\
0\\
1\\
1\\
0\\
1\\
0\\
\end{smallmatrix}\right)\right\}\\
&=\left\{\left(\begin{smallmatrix}
1\\
1\\
0\\
0\\
1\\
0\\
1\\
\end{smallmatrix}\right),
\left(\begin{smallmatrix}
1\\
1\\
1\\
1\\
1\\
1\\
1\\
\end{smallmatrix}\right)\right\},
\end{align*}
which contains the codeword $\bm{x}$.
\IEEEtriggeratref{11}

\newpage

\end{document}